\documentclass{article}
\usepackage{kr}

\usepackage{times}
\usepackage{soul}
\usepackage{url}
\usepackage[hidelinks]{hyperref}
\usepackage[utf8]{inputenc}
\usepackage[small]{caption}
\usepackage{algorithmic}
\usepackage{marvosym}
\usepackage{amsmath,amssymb, amsthm}
\usepackage[normalem]{ulem}
\usepackage{bm}
\usepackage{tikz}
\usepackage{pgfplots}
\usepackage{ltl}
\usepackage{graphicx}
\usepackage{todonotes}

\usepackage{subcaption} %
\usepackage[strings]{underscore}	 %
\usepackage{pifont}
\usepackage{cleveref}
\usepackage{xspace}
\usepackage[ruled,vlined]{algorithm2e}
\usepackage{mathtools}
\usetikzlibrary{arrows,automata,positioning,shapes,backgrounds,patterns,calc, fit}
\usepackage{xcolor}
\usepackage{multirow}
\usepackage{orcidlink}
\usepackage{listings}
\usepackage{comment}
\usepackage{booktabs}
\usepackage{wrapfig}
\usepackage{tree-dvips}
\usepackage{qtree}
\usepackage[export]{adjustbox}
\usepackage{thm-restate}

\newtheorem{theorem}{Theorem}
\newtheorem{definition}{Definition}
\newtheorem{proposition}{Proposition}
\newtheorem{corollary}{Corollary}

\allowdisplaybreaks[3]

\newcommand{\Buchi}{B\"{u}chi\xspace}
\newcommand{\bool}{\ensuremath{\mathbb{B}}}
\newcommand{\aut}{\ensuremath{\mathcal{A}}} %
\newcommand{\states}{\ensuremath{Q}}

\newcommand{\trace}{\ensuremath{t}}

\newcommand{\final}{\ensuremath{F}}
\newcommand{\pbf}[1]{\bool^{+}(#1)}

\newcommand\true{\mathit{true}}
\newcommand\false{\mathit{false}}
\newcommand{\from}{\ensuremath{\colon}}

\newcommand{\nats}{\ensuremath{\mathbb{N}}}
\newcommand{\lang}{\ensuremath{\mathcal{L}}}
\newcommand{\booleanformula}{\ensuremath{B}}

\newcommand{\tool}{\textsc{Compilaa}}
\newcommand{\autohyperq}{\textsc{AutoHyperQ}}
\newcommand{\goal}{\textsc{GOAL}}

\newcommand{\twopartelse}[3]
{
	\left\{
		\begin{array}{ll}
			#1 & \mbox{if } #2 \\
			#3 & \mbox{else} 
		\end{array}
	\right.
}

\newcommand{\sexists}{\exists^\circ} %
\newcommand{\sforall}{\forall^\circ}
\newcommand{\minsat}{\mathit{mSat}}

\definecolor{dkcyan}{rgb}{0.1, 0.3, 0.3}
\definecolor{dkgreen}{rgb}{0,0.3,0}
\definecolor{olive}{rgb}{0.5, 0.5, 0.0}
\definecolor{dkblue}{rgb}{0,0.1,0.5}

\definecolor{col:ln}{rgb}  {0.1, 0.1, 0.7}
\definecolor{col:str}{rgb} {0.8, 0.0, 0.0}
\definecolor{col:db}{rgb}  {0.9, 0.5, 0.0}
\definecolor{col:ours}{rgb}{0.0, 0.7, 0.0}

\definecolor{lightgreen}{RGB}{170, 255, 220}
\definecolor{darkbrown}{RGB}{121,37,0}

\colorlet{listing-comment}{gray}
\colorlet{operator-color}{darkbrown}
\colorlet{comment-color}{black!50}

\lstdefinelanguage{custom-lang}{
	keywords={let, in, match, with, when, if, then, else, elif, for, to, do, rec, return, new, not, and, while, any, goto, or},
	keywordstyle=[1]\color{dkblue}\bfseries,
	morekeywords=[2]{append, Set, Dict, Queue, pop, push, add, contains, Map, filter},
	keywordstyle=[2]\sffamily,
	morekeywords=[3]{conflictFixpoint,existsConflicts, universalConflicts, buildAutomaton, existentialKC, isEmpty, LTLtoASA, universalKC, existentialConfs, universalConfs},
	keywordstyle=[3]\color{dkcyan}\ttfamily,
	comment=[l][\color{comment-color}]{//},
	literate=%
	{=}{{{\color{operator-color}=}}}1
	{<-}{{{\color{operator-color}$\leftarrow$}}}1
	{|}{{{\color{dkblue}$\mid$}}}1
	{:}{{{\color{dkblue}:}}}1
	{:=}{{{\color{dkblue}:=}}}1
	{@}{ }1
}

\lstdefinestyle{default}{
	escapeinside={(*}{*)},
	basicstyle=,
	columns=fullflexible,
	commentstyle=\sffamily\color{black!50!white},
	framexleftmargin=1em,
	framexrightmargin=1ex,
	keepspaces=true,
	keywordstyle=\color{dkblue},
	mathescape,
	numbers=left,
	numberblanklines=false,
	numbersep=0.5em,
	numberstyle=\relscale{0.75}\color{gray}\ttfamily,
	showstringspaces=true,
	stepnumber=1,
	xleftmargin=1.2em,
}

\lstnewenvironment{mycode}[1][]
{\small
	\lstset{
		style=default, 
		language=custom-lang,
		#1
	}
}
{}

\def\x{0.5}
\def\y{1}
\def\ylabel{.85}

\newcommand{\awa}{\mathcal{A}}

\newcommand{\ekc}{\mathcal{A}^\exists}
\newcommand{\akc}{\mathcal{A}^\forall}

\title{Knowledge Compilation for Quantification in Alternating Automata}

\author{%
S. Akshay$^1$\and
Alfredo Cantarella$^2$\and
Supratik Chakraborty$^1$\and \\
Bernd Finkbeiner$^{2,3}$ \and
Niklas Metzger$^2$\\
\affiliations
$^1$Indian Institute of Technology Bombay, Mumbai, India\\
$^2$CISPA Helmholtz Center for Information Security, Saarbr\"ucken, Germany\\
$^3$Technical University of Munich, Munich, Germany\\
\emails
\{akshayss, supratik\}@cse.iitb.ac.in,
\{first.last\}@cispa.de,
finkbeiner@cispa.de
}

\begin{document}

\maketitle

\begin{abstract}
   We present a knowledge compilation approach for existential and universal quantification in alternating automata.
    Knowledge compilation transforms models into normal forms with special properties that enable efficient answering of questions of interest.  For Boolean formulas, several normal forms that have proven effective for existential/universal quantification, and even for functional synthesis, have been studied in the literature.
    For infinite word automata, quantification is a fundamental operation in verification tasks such as QPTL satisfiability checking and HyperLTL model checking.
    Existing algorithms rely on nondeterministic infinite word automata, where existential projection can be efficiently performed state-wise, but universal projection requires complementation. Complementing nondeterministic infinite word automata, however, is expensive, making existing algorithms infeasible for automata in practice.
    Towards addressing this problem, we propose novel knowledge compilation techniques for existential and universal quantification on alternating safety automata.
    Our approach compiles alternating automata into normal forms where projection can be applied uniformly and efficiently to each state's transition function.
    Using the compilations for each type of quantification, we can effectively eliminate a sequence of alternating quantifiers in formulas without complementation.
    Our BDD-based prototype demonstrates the practical effectiveness of our algorithms on a suite of QPTL satisfiability benchmarks.
\end{abstract}

\section{Introduction}

Linear Temporal Logic (LTL) has served as a central specification formalism for reasoning about programs since its introduction half a century ago~\cite{DBLP:conf/focs/Pnueli77}. Dominant techniques for reasoning about LTL typically proceed by translating LTL formulas into automata over infinite words -- most notably nondeterministic B\"{u}chi automata~\cite{DBLP:VardiW86,DBLP:Vardi95,DBLP:conf/icalp/SistlaVW85} and alternating B\"{u}chi automata~\cite{DBLP:journals/tcs/MiyanoH84,DBLP:conf/cade/Vardi97} -- and then applying automata-theoretic reasoning. 
Many of these approaches have been realized in state-of-the-art tools~\cite{spot,DBLP:conf/cav/Duret-LutzRCRAS22,meyer2018strix,meyer2021modernising,RSCDLP_dissectingltlsynt22} that routinely scale to large benchmarks, as evidenced in competitions like SyntComp~\cite{DBLP:journals/corr/abs-2206-00251}. However, for other applications such as hyperproperty verification, LTL is not expressive enough, and it becomes necessary to add quantifiers over propositions or even traces. This need was recognized early with the definition of Quantified Propositional Temporal Logic (QPTL) in~\cite{Sistla-QPTL,DBLP:conf/icalp/SistlaVW85}, which extends LTL with existential and universal quantifiers over propositions.
While complexity-theoretic studies  of QPTL were done in early work~\cite{Sistla-QPTL,DBLP:conf/icalp/SistlaVW85}, the development of practically efficient algorithms for QPTL satisfiabilility and model checking has not received much attention over the years.

The design of practical algorithms for reasoning about expressive logics like QPTL is often a balancing act between representational succinctness of the problem instance, and algorithmic tractability. Consider the problem of existentially quantifying an atomic proposition $p$ from an LTL formula $\varphi$.  %
If we convert $\varphi$ to a Nondeterministic B\"{u}chi Automaton (NBA), existentially quantifying $p$ is easy: simply project the $p$-labels from all state transitions. However, this comes at the cost of representation size, since translating LTL formulas to NBA often incurs an exponential blowup in state space.  For universally quantifying $p$, we must also complement the NBA for $\neg \varphi$ after projecting $p$-labels from its transitions. Unfortunately, complementing NBA is a computationally expensive operation~\cite{DBLP:conf/icalp/SistlaVW85,TFVT2011}, even for state-of-the-art tools~\cite{DBLP:conf/cav/Duret-LutzRCRAS22}. Alternating Büchi Automata (ABA) offer a compelling alternative that addresses the representation challenge. ABA generalize nondeterminism by adding universality, allowing the automaton to switch modes: While reading a single input word, it can nondeterministically choose a successor, or it can force the run into multiple parallel paths that must all accept the remainder of the input. ABA are exponentially more succinct than NBA, and support linear-time complementation~\cite{DBLP:journals/tcs/MiyanoH84}. However, existential quantification is much harder in ABA, since the run is a tree and distinct states reading the same input %
may impose contradictory requirements on a proposition. 
For example, for the LTL formula $\varphi$ represented by the ABA in Figure~\ref{fig:running:example:exists}, we cannot represent $\exists a.\varphi$ by simply removing $a$-leaves from the ABA.

Algorithms for QPTL satisfiability and model checking that start from ABA representations of LTL formulas typically proceed by converting the ABA into an NBA, thereby reintroducing the exponential state explosion and sacrificing the very succinctness that makes ABA attractive in the first place. More importantly, handling QPTL formulas with quantifier alternations forces one to complement the automaton at each alternation depth, leading to a nonelementary blow-up and rendering such approaches infeasible on complex formulas. This motivates us to ask whether one can represent ABA in a form that supports quantification directly and efficiently, while remaining within the alternating setting, and avoiding the intermediate nonelementary explosion. While the final emptiness check for ABA remains computationally demanding, eliminating repeated expensive complementation shifts the complexity away from this limiting step, making previously intractable formulas accessible.

\begin{figure}[t]
\centering
    \begin{subfigure}[b]{.53\linewidth}
    \resizebox{.85\textwidth}{!}{%
        \tikzstyle{state}=[draw, circle, fill=none, minimum width=0.7cm, 
minimum height = 0.7cm,
align=center, thick]
\tikzstyle{pseudostate}=[draw, circle, fill, minimum size=1.5pt, inner sep=0pt, outer sep=0pt, fill]

\begin{tikzpicture}[>=stealth',auto, minimum size=0pt, inner sep=0pt, outer sep=0pt]

\node[state, accepting] (p0){$q_0$};
\node (p0label) [above right = 0.3*\ylabel and 0.1*\x of p0]{$\LTLglobally (a \wedge \LTLnext (\neg a \vee b))$};

\node[state, draw=none] (p1) [below left = \y and \x of p0]{$a$};

\node[state, accepting, below right = \y and \x of p0] (p2) {$q_1$};
\node (p2label) [above right = 0.3*\ylabel and 0.1*\x of p2]{$\neg a \vee b$};

\node[state, draw = none] (p3) [below left = \y and \x of p2]{$\neg a$};

\node[state, draw=none] (p4) [below right = \y and \x of p2]{$b$};

\node[state, draw=none] (p5) [below left = \y and \x of p2]{$\neg a$};

\path 
(p0) edge[->,thick] coordinate[pos=0.25](p0l) node[above ] {%
      } (p1)
(p0) edge[->,thick] coordinate[pos=0.25](p0r) node[above ] {%
      } (p2)
(p0) edge[->,thick, loop left] coordinate[pos=0.25](p0loop) node[above ] {%
      } (p0)

(p2) edge[->,thick] node[above ] {%
      } (p3)
(p2) edge[->,thick]  node[above ] {%
      } (p4)

(p0) edge[dotted] node[above ] {%
      } (p0label)
(p2) edge[dotted] node[above ] {%
      } (p2label)

(p0l) edge[-,bend right, thick, inner sep = 0pt] (p0r)
(p0loop) edge[-,bend right, thick, inner sep = 0pt] (p0l)

;
\end{tikzpicture}
    }
        \caption{}
        \label{fig:running:example:exists}
    \end{subfigure}
    \hfill
    \begin{subfigure}[b]{.46\linewidth}
        \resizebox{0.82\textwidth}{!}{%
        \tikzstyle{state}=[draw, circle, fill=none, minimum width=0.7cm, 
minimum height = 0.7cm,
align=center, thick]
\tikzstyle{pseudostate}=[draw, circle, fill, minimum size=1.5pt, inner sep=0pt, outer sep=0pt, fill, show background rectangle]

\begin{tikzpicture}[>=stealth',auto, minimum size=0pt, inner sep=0pt, outer sep=0pt]

\node[state, accepting] (p0){$q_0$};

\node[state, draw=none] (p1) [below left = \y and \x of p0]{$a$};

\node[state, accepting, minimum width = 0.95cm, below right = \y and \x of p0] (p2) {$q_0, q_1$};

\node[state, draw = none] (p3) [below left = \y and \x of p2]{$a$};

\node[state, draw=none] (p4) [below right = 1.4\y and 1*\x of p2]{$b$};

\node[state, draw=none] (p5) [below right = 1.4*\y and -1.8*\x of p2]{$\neg a$};

\node[circle] (p6) [below right = 0.6*\y and 0.1*\x of p2]{};

\path 
(p0) edge[->,thick] coordinate[pos=0.28](p0l) node[above ] {%
      } (p1)
(p0) edge[->,thick] coordinate[pos=0.3](p0r) node[above ] {%
      } (p2)

(p2) edge[->,thick] coordinate[pos=0.25](p0l1) node[above ] {%
      } (p3)
(p2) edge[->,thick, out=198, in=162, min distance=0.8cm] coordinate[pos=0.15](p0loop1) node[above ] {%
      } (p2)
(p2) edge[-,thick] coordinate[pos=0.5](p2useless) node[above ] {%
      } (p6)

(p6) edge[->,thick] coordinate[pos=0.25](p0r1) node[above ] {%
      } (p5)
(p6) edge[->,thick] coordinate[pos=0.25](p0r1) node[above ] {%
      } (p4)

(p0l1) edge[-,bend right=15, thick, inner sep = 0pt] (p2useless)
(p0loop1) edge[-,bend right=15, thick, inner sep = 0pt] (p0l1)

(p0l) edge[-,bend right=20, thick, inner sep = 0pt] (p0r)

;
\end{tikzpicture}
    }
        \caption{}
        \label{fig:running:example:exists:compiled}
    \end{subfigure}
    \caption{(a) An alternating %
    automaton for the LTL formula $\varphi = \protect\LTLglobally (a \wedge \protect\LTLnext (\neg a \vee \protect b))$. (b) The compiled automaton for the existential quantification of the QPTL formula $\exists a. \varphi$. States $q_0$ and $q_1$ are in conflict and merged whenever they are reached conjunctively.}
\end{figure}
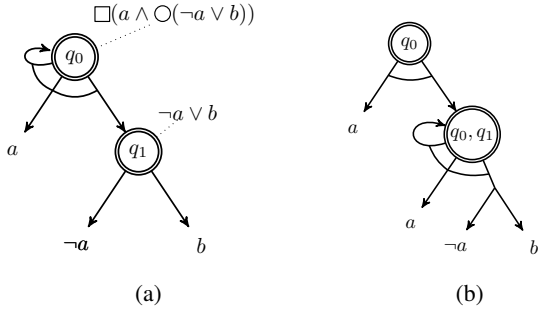
\medskip 

In this paper, we propose a solution for the above challenge via knowledge compilation. Our work is inspired by the use of normal forms for existential quantification in the setting of Boolean models. Such forms include circuits in decomposable negation normal form (DNNF)~\cite{darwiche-jacm,ddnnf}, weak-DNNF~\cite{DBLP:conf/cav/AkshayCGKS18,DBLP:conf/aaai/IllnerK24}, synthesis negation normal form (SynNNF) and variants~\cite{DBLP:conf/lics/ShahBAC21,DBLP:journals/amai/AkshayCS24}, for which practically efficient knowledge compilation algorithms have been studied. Knowledge compilation for interval automata, which are generalizations of Binary Decision Diagrams~\cite{Bryant86}, has been studied in~\cite{DBLP:interval-automata-2010}; however, interval automata are unrelated to automata over infinite words. 

Our main novelty lies in extending the knowledge compilation framework for Boolean models to the domain of infinite word automata. Rather than performing a full ABA to NBA translation, we propose a technique to detect and repair only the specific states that cause problems or {\em conflicts} during quantification. Our approach compiles an alternating automaton into a conflict-free representation where existential projection can be applied \emph{state-wise} to the transition function, preserving the compact alternating structure. Consider again the example in Figure~\ref{fig:running:example:exists}, where one possible ABA for the LTL formula $\varphi = \protect\LTLglobally (a \wedge \protect\LTLnext (\neg a \vee \protect b))$ is shown. The compiled automaton for the QPTL formula $\exists a. \varphi$ is shown in~Figure~\ref{fig:running:example:exists:compiled}. Significantly, the compiled automaton is an ABA language-equivalent to the given ABA, and allows existentially quantifying $a$ by simply quantifying $a$ from each state's transition function.
Similarly, in Figure~\ref{fig:running:example:universal} and~\ref{fig:running:example:universal:compiled}, we show an ABA and its language-equivalent compiled form for universally quantifying $a$.  The ABA corresponding to the quantified formula is obtained by universally quantifying $a$ from each state's transition function in Figure~\ref{fig:running:example:universal:compiled}.  

We summarize our primary technical contributions below.
\begin{itemize}
\item We define knowledge representations for existential and universal quantification for ABA, and develop syntactic conditions for these forms.
\item For alternating safety automata (ASA), we identify the notion of \emph{conflicts} for existential quantification, and provide a fixpoint-based compilation algorithm that identifies the set of conflicting states. We also present a repair algorithm based on a \emph{modified Miyano-Hayashi construction}~\cite{DBLP:journals/tcs/MiyanoH84} that resolves these conflicts locally.
\item We extend our conflict detection and repair mechanisms to support \emph{universal quantification}, enabling practically efficient satisfiability checking for safety QPTL formulas with quantifier alternations.
    \item We implement our approach in a BDD-based prototype and evaluate it on \emph{QPTL satisfiability} benchmarks against two state-of-the-art QPTL solvers. We demonstrate significant performance improvements by avoiding complementation operations via knowledge compilation. 
    
\end{itemize}

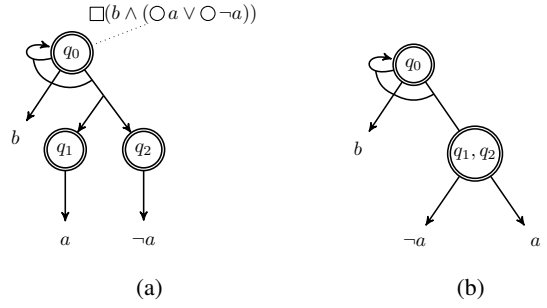
\begin{figure}[t]
\centering
    \begin{subfigure}[b]{.53\linewidth}
    \resizebox{.85\textwidth}{!}{
        \tikzstyle{state}=[draw, circle, fill=none, minimum width=0.7cm, 
minimum height = 0.7cm,
align=center, thick]
\tikzstyle{pseudostate}=[draw, circle, fill, minimum size=1.5pt, inner sep=0pt, outer sep=0pt, fill]

\begin{tikzpicture}[>=stealth',auto, minimum size=0pt, inner sep=0pt, outer sep=0pt]

\node[state, accepting] (p0){$q_0$};
\node (p0label) [above right = 0.3*\ylabel and 0.1*\x of p0]{$\LTLglobally (b \wedge (\LTLnext a \vee \LTLnext \neg a))$};

\node[state, draw=none] (p1) [below left = \y and \x of p0]{$b$};

\node[state, accepting, below right = 1.2*\y and 1.5*\x of p0] (p3) {$q_2$};

\node[state, accepting, below right = 1.2*\y and -1.3* \x of p0] (p2) {$q_1$};

\node[state, draw = none] (p4) [below = 0.9*\y of p3]{$\neg a$};

\node[state, draw=none] (p6) [below = 0.9\y of p2]{$a$};

\node[circle] (p7) [below right = 0.5*\y and 0.6*\x of p0]{};

\path 
(p0) edge[->,thick] coordinate[pos=0.25](p0l) node[above ] {%
      } (p1)
(p0) edge[-,thick] coordinate[pos=0.4](p0r) node[above ] {%
      } (p7)
(p0) edge[->,thick, loop left] coordinate[pos=0.25](p0loop) node[above ] {%
      } (p6)

(p7) edge[->,thick] node[above ] {%
      } (p2)
(p7) edge[->,thick]  node[above ] {%
      } (p3)

(p2) edge[->,thick]node[above ] {%
      } (p6)
(p3) edge[->,thick]node[above ] {%
      } (p4)

(p0) edge[dotted] node[above ] {%
      } (p0label)

(p0l) edge[-,bend right, thick, inner sep = 0pt] (p0r)
(p0loop) edge[-,bend right, thick, inner sep = 0pt] (p0l)

;
\end{tikzpicture}
    }
        \caption{}
        \label{fig:running:example:universal}
    \end{subfigure}
    \hfill
    \begin{subfigure}[b]{.46\linewidth}
        \resizebox{0.82\textwidth}{!}{
        \tikzstyle{state}=[draw, circle, fill=none, minimum width=0.7cm, 
minimum height = 0.7cm,
align=center, thick]
\tikzstyle{pseudostate}=[draw, circle, fill, minimum size=1.5pt, inner sep=0pt, outer sep=0pt, fill]

\begin{tikzpicture}[>=stealth',auto, minimum size=0pt, inner sep=0pt, outer sep=0pt]

\node[state, accepting] (p0){$q_0$};

\node[state, draw=none] (p1) [below left = \y and \x of p0]{$b$};

\node[state, accepting, minimum width = 0.95cm, below right = \y and \x of p0] (p2) {$q_1, q_2$};

\node[state, draw = none] (p4) [below left = \y and \x of p2]{$\neg a$};

\node[state, draw=none] (p6) [below right = \y and \x of p2]{$a$};

\path 
(p0) edge[->,thick] coordinate[pos=0.25](p0l) node[above ] {%
      } (p1)
(p0) edge[-,thick] coordinate[pos=0.25](p0r) node[above ] {%
      } (p2)
(p0) edge[->,thick, loop left] coordinate[pos=0.25](p0loop) node[above ] {%
      } (p6)

(p2) edge[->,thick]node[above ] {%
      } (p6)
(p2) edge[->,thick]node[above ] {%
      } (p4)

(p0l) edge[-,bend right, thick, inner sep = 0pt] (p0r)
(p0loop) edge[-,bend right, thick, inner sep = 0pt] (p0l)

;
\end{tikzpicture}
    }
        \caption{}
        \label{fig:running:example:universal:compiled}
    \end{subfigure}
    \caption{(a) An alternating 
    automaton for the LTL formula $\varphi = \protect\LTLglobally (b \wedge (\protect \LTLnext a \vee \protect \LTLnext \neg a))$. (b) The compiled automaton for the universal quantification of the QPTL formula $\forall a. \varphi$ . States $q_1$ and $q_2$ are in conflict.}
\end{figure}

\section{Preliminaries}\label{sec:preliminaries}
We consider a finite set of Boolean variables $\Sigma{}$. 
An infinite trace $\trace \subseteq (2^{\Sigma})^\omega$ is an infinite sequence of subsets of Boolean variables.
We refer to the $i$'th element of the trace by $\trace[i]$.
An (acyclic) tree $\mathcal{T}$ over a set of directions $D$ is a prefix-closed subset of $D^*$.
The empty sequence $\epsilon$ is the root, and the children of a node $n \in \mathcal{T}$ are the nodes $\mathit{children}(n) = \{n \cdot d \in \mathcal{T} \mid d \in D\}$. 
A $\Sigma$-labeled tree is a pair $(\mathcal{T}, r)$, where $r : \mathcal{T} \rightarrow \Sigma$ is the labeling function.
A directed acyclic graph (DAG) $(V,E)$ consists of a set of vertices $V$ and a set of edges $E \subseteq V \times V$.
We define the set of successors of a vertex $v$ in a DAG $(V, E)$ as $\mathit{post}(v) = \{v' \in V \mid (v, v') \in E \}$.
The set of reachable vertices for a vertex $v$ then is $\mathit{reach}(v) = \{v' \in V \mid \exists v_0v_1\ldots v_n \text{ s.t. } (v_i, v_{i+1}) \in E, v_0 = v, \text{ and } v_n = v'\}$.
We assume DAGs to be acyclic.

\textbf{Boolean formulas.}
A Boolean formula over a set of Boolean variables $\Sigma$ is a combination of the standard operators $\wedge$ (and), $\vee$ (or), $\neg$ (negation), and the derived operators $\rightarrow$ (implication), and $\leftrightarrow$ (equivalence).
The semantics of a Boolean formula $\booleanformula(\Sigma)$ over a set of Boolean variables $\Sigma$ is a mapping $\booleanformula : 2^{\Sigma} \rightarrow \{0,1\}$.
A subset $X \subseteq 2^{\Sigma}$ is a satisfying assignment of $\booleanformula$ if $\booleanformula(X)$ evaluates to $1$.
We denote the set of all satisfying assignments of $B$ as $\mathit{Sat(B)}$.
The set of minimal satisfying assignments is $\minsat(B) = \{X \in Sat(B) \mid \forall X' \subset X. X' \notin \mathit{Sat}(B)\}$.
We furthermore use $X \vDash B$ if $X \in \minsat(B)$ and $X \nvDash B$ if $X \notin \minsat(B)$.
Setting a variable $a \in \Sigma$ in a Boolean formula $B$ to either $\true$ or $\false$ is denoted by $\booleanformula[a \mapsto 1]$ and $\booleanformula[a \mapsto 0]$.
Setting a subset of variables $A \subset \Sigma$ in $B$ is denoted by $B[A \rightarrow 1]$ and $B[A\rightarrow 0]$.
Existential quantification of a variable $a \in \Sigma$ in a Boolean formula is the Boolean formula $\exists a. \booleanformula = \booleanformula[a \mapsto 1] \vee \booleanformula[a \mapsto 0]$.
Universal quantification of a variable $a \in \Sigma$ in a Boolean formula is the Boolean formula $\forall a. \booleanformula = \booleanformula[a \mapsto 1] \wedge \booleanformula[a \mapsto 0]$.
We denote the set of \emph{positive} Boolean formulas over a set of variables $\Sigma$ as $\mathbb{B}^+(\Sigma)$, whereby a Boolean formula is positive if $\neg$ does not occur in its representation.

\textbf{Alternating Automata.}
The set of literals $L_{\Sigma} = \{\neg a ~|~ a \in \Sigma\} \cup \Sigma$ for a set of Boolean variables contains all positive and negative literals of $\Sigma$.
We use elements of $L_{\Sigma}$ as shorthand for truth assignments to variables in Boolean formulas, and use all notations accordingly. 
An \emph{alternating \Buchi automaton} $ \aut $ over Boolean variables $\Sigma$ is a tuple $\aut = (\Sigma,\states,q_0,\delta,\final)$, where: $\states$ is a finite set of states, $q_{0} \in \states$ is the initial state, $\delta \from \states \to \pbf{\states{}\cup\ L_{\Sigma}} $ is the \emph{transition function}, and $\final \subseteq \states$ is the \Buchi accepting condition.
We define the set of direct successors for a state $q$ as $\mathit{post}(q) = \{X \cap Q \mid X \in \minsat(\delta(q))\}$.
Similarly, we define the set of literals appearing in the transition function of a state $q$ as $l(q) = \{a \in L_{\Sigma} \mid \exists X \text{ s.t. } X \vDash \delta(q) \text{ and } a \in X\}$.
The set of all reachable states for a state $q$ is defined as $\mathit{reach}(q) = \{q' \in \states \mid \exists q_0q_1\ldots q_n \text{ s.t. } q_0 = q, q_n = q' \text{ and } q_{i+1} \in \bigcup post(q_i)\}$.
A \emph{run} of an alternating automaton on a word $\alpha \in (2^\Sigma)^\omega$ is a ${Q \cup L_{\Sigma}}$-labeled tree $(\mathcal{T}, r)$, where $r(\epsilon) = q_0$ and for all $n \in \mathcal{T}$, if $r(n) = X$, then $\{r(n') \mid n' \in \mathit{children(n)}\}$ satisfies $\delta(q)$, where $q = X|_Q$, and $\alpha[\vert n \vert] = X|_{L_{\Sigma}}$.
A run is \emph{accepting} if for all infinite branches $n_0n_1\ldots$ of the run tree, the sequence $r(n_0)r(n_1)\ldots$ visits states in $F$ infinitely often.
A B\"uchi automaton is \emph{nondeterministic}, if for every state $q \in Q$, the transition function $\delta(q)$, after projection to $Q$, is a disjunction of states.
A B\"uchi automaton is \emph{universal}, if for every state $q \in Q$, the transition function $\delta(q)$, after projection to $Q$, is a conjunction of states.
An alternating automaton is a \emph{safety} automaton if $\final = Q$.
For a trace $\trace \in (2^\Sigma)^\omega$, the \emph{projection} of $\trace$ onto $\Sigma' \subseteq \Sigma$ is defined as $\trace|_{\Sigma'} = (\trace[0] \cap \Sigma')(\trace[1] \cap \Sigma') \cdots$.
The \emph{existential quantification} of an alternating automaton $\aut$ over $\Sigma$ for variable $p \in \Sigma$ is the automaton $\exists p. \aut$ over alphabet $\Sigma \setminus \{p\}$ such that $\mathcal{L}(\exists p. \aut) = \{\trace|_{\Sigma \setminus \{p\}} \mid \exists \trace' \in (2^\Sigma)^\omega. \trace'|_{\Sigma \setminus \{p\}} = \trace \text{ and } \trace' \in \mathcal{L}(\aut)\}$.
The \emph{universal quantification} of $\aut$ for a variable $p \in \Sigma$ is the automaton $\forall p. \aut$ over alphabet $\Sigma \setminus \{p\}$ such that $\mathcal{L}(\forall p. \aut) = \{\trace|_{\Sigma \setminus \{p\}} \mid \forall \trace' \in (2^\Sigma)^\omega. \trace'|_{\Sigma \setminus \{p\}} = \trace \text{ implies } \trace' \in \mathcal{L}(\aut)\}$.

\textbf{QPTL.} Quantified Propositional Temporal Logic (QPTL) \cite{Sistla-QPTL} extends LTL \cite{DBLP:conf/focs/Pnueli77} with quantification over propositional variables. QPTL formulas over a set $\Sigma{}$ of Boolean variables are generated by the following grammar, whereby $a \in \Sigma$:
$ \varphi ::= a \mid \neg \varphi \mid \varphi \land \varphi \mid \LTLnext \varphi \mid \varphi \LTLuntil \varphi \mid \exists p. \varphi $.
The formula $a$ states that the Boolean variable $a$ needs to hold at the current step of a trace, $\LTLnext \varphi$ requires $\varphi$ to hold in the \emph{next} step and $\varphi_1 \LTLuntil \varphi_2$ requires  $\varphi_1$ to hold \emph{until} $\varphi_2$ eventually holds. We derive the usual Boolean constants and connectives $\true, \false, \lor, \rightarrow, \leftrightarrow$, as well as the temporal operators \emph{eventually} ($\LTLfinally \varphi := \true \LTLuntil \varphi$) and \emph{globally} ($\LTLglobally \varphi := \neg \LTLfinally \neg \varphi$).
In addition, we use the \emph{universal quantification} ($\forall p. \varphi := \neg \exists p. \neg \varphi $).
The LTL part of the formula follows the standard LTL semantics~\cite{DBLP:conf/focs/Pnueli77}.
The existential operator states that there must exist an infinite sequence of assignments of Boolean values to p s.t.\ the remaining formula holds. We refer to~\cite{Sistla-QPTL} for details.
Translating the quantifier-free LTL sub-formula of a QPTL formula written in prenex normal form to an alternating automaton can be done in linear time \cite{DBLP:conf/cade/Vardi97}.

\section{Normal Forms for Existential and Universal Quantification}\label{sec:knowledge:compilations:for:existential:and:universal:quantification}
The goal of this paper is to construct alternating automata on which existential and universal quantification can be performed state-wise.
We begin by defining the quantification of variables on \emph{alternating} automata in a state-wise manner, analogous to how existential quantification is performed on nondeterministic automata via projection.
In nondeterministic automata, existential quantification of a variable $p$ can be applied locally to each state's transition function by simply projecting $p$ from the alphabet and existentially quantifying it in the transition formula.
This local operation is sound because nondeterministic runs follow a single path through the automaton.
We define state-wise existential quantification for alternating automata formally:

\begin{definition}[State-wise Existential Quantification]\label{def:state:wise:existential:projection}
Let $\awa = (\Sigma, \states, q_0, \delta, F)$ and let $p \in \Sigma$. 
The \emph{state-wise existential quantification} $\sexists p. \awa$ is the automaton $\awa' = (\Sigma', \states, q_0, \delta', F)$, where 
$\Sigma' = \Sigma \setminus \{p\}$ and $\delta'(q) = \exists p. \delta(q)$.
\end{definition}
State-wise existential quantification applies existential quantification locally to each state's transition function.
This operation removes the variable $p$ from the alphabet and quantifies it in each transition formula independently.
The universal case is defined analogously, where we apply universal quantification to each state's transition function:

\begin{definition}[State-wise Universal Quantification]\label{def:state:wise:universal:projection}
Let $\awa = (\Sigma, \states, q_0, \delta, F)$ and let $p \in \Sigma$. 
The \emph{state-wise universal quantification} $\sforall p. \awa$ is the automaton $\awa' = (\Sigma', \states, q_0, \delta', F)$, where 
$\Sigma' = \Sigma \setminus \{p\}$ and $\delta'(q) = \forall p. \delta(q)$.
\end{definition}

Applying state-wise quantification directly to an arbitrary alternating automaton can dramatically alter its language.
To illustrate this, consider our running example from \Cref{fig:running:example:exists}.
The automaton represents the LTL formula $\varphi = \LTLglobally (a \wedge \LTLnext (\neg a \vee b))$, which requires that $a$ holds at every position, and at the next position either $\neg a$ or $b$ must hold.
For the QPTL formula $\exists a. \varphi$, we want to check if there exists an infinite assignment to $a$ such that the formula is satisfied.

If we naively apply state-wise existential quantification to remove $a$ from the automaton, the operation assumes that every evaluation of $a$ was part of the language.
This is because the transitions $\delta(q_0) = a \wedge q_1$ and $\delta(q_1) = (\neg a  \vee b)$ become $\delta'(q_0) = q_1$ and $\delta'(q_1) = \top$ after existentially quantifying $a$.
This completely ignores the fact that only $a$ being $\true$ at every position satisfies the original formula while existentially quantifying $a$ in $\delta(q_1)$.
This can result in the automaton accepting a superset of the original language, violating language preservation.
We now define normal forms for alternating automata where state-wise existential and universal quantifications are language preserving.

\begin{definition}[Existential Normal Form]\label{def:existential:knowledge:compilation}
    Let $\awa$ be an alternating automaton over $\Sigma$. 
    An existential normal form of $\awa$ for the Boolean variable $p \in \Sigma$ is an alternating automaton $\ekc$ such that $\mathcal{L}(\sexists p.\ekc)$ is equivalent to $\mathcal{L}(\exists p.\awa)$.
\end{definition}
An automaton is in existential normal form for a variable $p$ if applying state-wise existential quantification to $p$ preserves the language of the automaton. 
We denote the operator $\exists p. \awa$ as defined in \Cref{sec:preliminaries}.
The universal normal form is defined analogously: We require that state-wise \emph{universal} quantification preserves the language of the automaton.

\begin{definition}[Universal Normal Form]\label{def:universal:knowledge:compilation}
    Let $\awa$ be an alternating automaton over $\Sigma$. 
    A universal normal form of $\awa$ for the Boolean variable $p \in \Sigma$ is an alternating automaton $\akc$ such that $\mathcal{L}(\sforall p. \akc)$ is equivalent to $\mathcal{L}(\forall p. \awa)$.
\end{definition}

Before continuing with formal requirements for existential and universal normal forms, we highlight an observation for the automata usually used in verification algorithms.

\begin{proposition}\label{prop:nondeterministic:is:kc}
Every nondeterministic automaton is in existential normal form.
Every universal automaton is in universal normal form.
Every deterministic automaton is in both.
\end{proposition}
For nondeterministic automata, an accepting run on an input word $w$ is a single path through the automaton. 
State-wise existential quantification preserves the language because each state along this path can independently choose a valuation for $p$ that satisfies its transition constraint.
For universal automata, an accepting run on an input word $w$ simultaneously visits all states reachable under the universal transitions. 
State-wise universal quantification preserves the language because each state must provide $p$ and $\neg p$.
For deterministic automata, an accepting run on an input word $w$ is a unique path, combining the cases of nondeterministic and universal automata.

\section{Syntactic Conditions for Knowledge Representation}\label{sec:conditions:for:knowledge:compilations}

In this section, we establish formal conditions that characterize when an alternating automaton is in existential or universal normal form.
These conditions form the foundation for the knowledge compilation algorithms presented in~\Cref{sec:computing:knowledge:compilations:for:safety:automata}.
Rather than operating on the automaton's individual runs, our conditions are defined on \emph{partial unfoldings} of the automaton.
Partial unfoldings follow the transition function while resolving either nondeterministic or universal choices.
These partial unfoldings provide a way to reason about potential conflicts in the automaton's structure without fully exploring all possible input words.
For both existential and universal quantification, we define corresponding existential and universal unfoldings and demonstrate how they capture sufficient conditions for the respective normal forms.

\begin{figure}
    \begin{subfigure}[b]{.56\linewidth}
    \resizebox{.9\textwidth}{!}{
        \tikzstyle{state}=[draw = none,  fill=none, minimum width=0.7cm, 
minimum height = 0.7cm,
align=center, thick]
\tikzstyle{pseudostate}=[draw, fill, minimum size=1.5pt, inner sep=0pt, outer sep=0pt, fill, show background rectangle]

\begin{tikzpicture}[>=stealth',auto, minimum size=0pt, inner sep=0pt, outer sep=0pt]

\node[state, draw=none] (p0){$q_0, 0$};

\node[state, draw=none] (p1) [below left = \y and 1.5*\x of p0]{$a, 1$};

\node[state] (p2) [ below = \y  of p0] {$q_0, 1$};

\node[state] (p3) [ below right = \y and 1.5*\x of p0] {$q_1, 1$};

\node[state, draw = none] (p4) [below right = \y and \x of p3]{$\neg a, 2$};

\node[state, draw=none] (p5) [below = \y of p1]{$a, 2$};

\node[state] (p6) [ below = \y  of p2] {$q_0, 2$};

\node[state] (p7) [ below = \y of p3] {$q_1, 2$};

\node[coordinate] (p8) [below = 0.3*\y of p6]{};

\node[coordinate] (p9) [below = 0.3*\y of p7]{};

\path 
(p0) edge[->,thick] coordinate[pos=0.25](p0l) node[above ] {%
      } (p1)
(p0) edge[->,thick] coordinate[pos=0.25](p0r) node[above ] {%
      } (p2)
(p0) edge[->,thick] coordinate[pos=0.25](p0r) node[above ] {%
      } (p3)

(p2) edge[->,thick] node[above ] {%
      } (p5)
(p2) edge[->,thick] node[above ] {%
      } (p6)
(p2) edge[->,thick] node[above ] {%
      } (p7)

(p3) edge[->,thick] node[above ] {%
      } (p4)

(p6) edge[-, thick, dotted] node[above ] {%
      } (p8)
(p7) edge[-,thick, dotted] node[above ] {%
      } (p9)
;
\end{tikzpicture}
    }
        \caption{}
        \label{fig:existential:unfolding}
    \end{subfigure}
    \hfill
    \begin{subfigure}[b]{.43\linewidth}
        \resizebox{0.9\textwidth}{!}{
        \tikzstyle{state}=[draw = none,  fill=none, minimum width=0.7cm, 
minimum height = 0.7cm,
align=center, thick]
\tikzstyle{pseudostate}=[draw, fill, minimum size=1.5pt, inner sep=0pt, outer sep=0pt, fill, show background rectangle]

\begin{tikzpicture}[>=stealth',auto, minimum size=0pt, inner sep=0pt, outer sep=0pt]

\node[state, draw=none] (p0){$q_0, 0$};

\node[state, draw=none] (p1) [below left = \y and 1.5*\x of p0]{$b, 1$};

\node[state] (p2) [ below = \y  of p0] {$q_1, 1$};

\node[state] (p3) [ below right = \y and 1.5*\x of p0] {$q_2, 1$};

\node[state, draw = none] (p4) [below = \y of p2]{$ a, 2$};

\node[state, draw = none] (p42) [below = \y of p3]{$\neg a, 2$};

\path 
(p0) edge[->,thick] coordinate[pos=0.25](p0l) node[above ] {%
      } (p1)
(p0) edge[->,thick] coordinate[pos=0.25](p0r) node[above ] {%
      } (p2)
(p0) edge[->,thick] coordinate[pos=0.25](p0r) node[above ] {%
      } (p3)

(p2) edge[->,thick] node[above ] {%
      } (p4)

(p3) edge[->,thick] node[above ] {%
      } (p42)

;
\end{tikzpicture}
    }
        \caption{}
        \label{fig:universal:unfolding}
    \end{subfigure}
    \caption{We sketch an existential (a) and universal unfolding (b) of the automata in \Cref{fig:running:example:exists} and \Cref{fig:running:example:universal}, respectively.}
\end{figure}
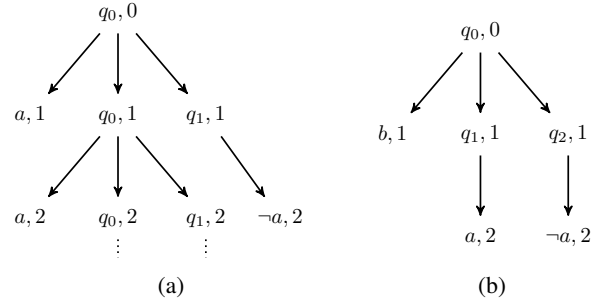
\subsection{Condition for Existential Normal Forms}\label{sec:syntactic:condition:existential}

The syntactic condition for existential normal forms is formalized on a special unfolding of the automaton, called \emph{existential unfolding}.
The general idea is to follow the transition function of the automaton, but resolving nondeterministic choices, i.e., whenever the transition function yields more than one choice for sets of future states, the symbolic existential run chooses exactly one such set.
Formally, this yields an infinite acyclic graph:

\begin{definition} [Existential Unfolding]
 An \emph{existential unfolding} of an alternating \Buchi\ automaton $\aut$ is a DAG $(V,E)$, with  $ V \subseteq (\states \cup L_{\Sigma}) \times \nats $,  $ (q_0,0) \in V $ and\\
  $E \subseteq \bigcup_{i \in \mathbb N} \Big[ (\states \times \{i\}) \times
  (\states \times \{i+1\})\Big]$\\
  $\phantom{lllllllllllllllll}\cup \Big[ (\states \times \{i\}) \times (L_{\Sigma} \times \{i+1\})\Big],$\\
and $V,E$ are smallest sets s.t. $\forall (q,i) \in V, post(q,i) = X,$ 
\begin{align*}
  &-\  X \in \minsat(\delta(q)) \text{ and } (X \times \{i + 1\}) \subseteq V, \text{ and}\\
  &-\ \{(q, i)\} \times (X \times \{i+1\}) \subseteq E
\end{align*}
\end{definition}

The vertices of the DAG are all states and literals for every timestep.
We transition from timesteps to timesteps by choosing a set of future states and literals, such that (1) the transition function is satisfied for every state, and (2) the edges point to the correct future vertices.
Note that we choose the literals for the variables in $\Sigma$ for each state, but do not enforce that the variable assignments are unique over all states in every branch of the unfolding.
This is necessary to find states that require contradictory valuations for $p$ at the same timestep.
\Cref{fig:existential:unfolding} sketches an existential unfolding of the automaton in \Cref{fig:running:example:exists}.
Every state and literal is duplicated for every timestep, and the edges follow the transition function of the automaton.
The unfolding follows all states of the conjunction in $q_0$ and chooses one of the disjuncts in $q_1$.
Using existential unfoldings, we can define the \emph{syntactic condition} for existential quantification. 

\begin{definition}[Existential Syntactic Condition]\label{def:syntactic:exists:condition}
An alternating automaton $\awa$ satisfies the existential syntactic condition for $p \in \Sigma$ if for all its existential unfoldings $(V, E)$ and for all $i \in \mathbb{N}$ it holds that either $(p, i) \notin V$ or $(\neg p, i) \notin V$.
\end{definition}
This condition requires that no run contains both $p$ and $\neg p$ at the same timestep.
Since all nondeterminism is resolved in the existential unfolding, finding dual literals at the same timestep implies that the automaton traversed through a conjunction of states and arrived in two states where different evaluations for $p$ are required.
This happens at timepoint $2$ in the existential unfolding in \Cref{fig:existential:unfolding}, where the automaton is in both $q_1$ and $q_2$ at the same time, which requires $\neg a$ and $a$, respectively.
This is a violation of the syntactic condition, and therefore the automaton in \Cref{fig:running:example:exists} is not in existential normal form for $a$.
This is a critical situation for existential quantification: The state-wise quantification would yield positive results for each state, but the transition functions for both states could be contradictory.
In~\Cref{sec:computing:knowledge:compilations:for:safety:automata}, we present an algorithm that identifies such conflicts and merges the conflicting states, which results in an automaton that satisfies the syntactic condition.

\subsection{Condition for Universal Normal Forms}\label{sec:syntactic:condition:universal}
The universal case follows a dual structure. 
Rather than resolving nondeterministic choices as in existential unfoldings, universal unfoldings resolve \emph{universal} choices, ensuring that, whenever two states occur at the same timepoint, they are preceeded by nondeterminism. 
This captures the scenario where different runs of the automaton, arising from nondeterministic choices at earlier points, allow different evaluations of the same variable at the same timestep.
We define the universal unfolding accordingly:

\begin{definition}[Universal Unfolding]\label{def:universal:unfolding}
 A \emph{universal unfolding} of an alternating \Buchi\ automaton $\aut$ is a DAG $(V,E)$, with
  $ V \subseteq (\states \cup  L_{\Sigma}) \times \nats $ and $ (q_0,0) \in V $,
  where

  $E \subseteq \bigcup_{i \in \mathbb N} \Big[ (\states \times \{i\}) \times
  (\states \times \{i+1\})\Big]$\\
  $\phantom{llllllllllllllllllll}\cup \Big[ (\states \times \{i\}) \times (L_{\Sigma} \times \{i+1\})\Big],$\\
and $V,E$ are smallest sets s.t. $\forall (q,i) \in V, post(q,i)=X$, 
\begin{align*}
  &- X \subseteq Q \cup L_\Sigma \text{ is a minimal set s.t. } \forall X' \in \minsat(\delta(q)),\\ 
  & ~~~~~~~~~~X \cap X' \cap ~Q \neq \emptyset,~ X' \cap L_\Sigma \subseteq X, \text{ and} \\
  &- (X \times \{i + 1\}) \subseteq V\text{ and }
   \{(q, i)\} \times (X \times \{i+1\}) \subseteq E
    \end{align*}
\end{definition}
The universal unfolding follows a dual structure to the existential case: Instead of resolving nondeterministic choices, it resolves \emph{universal} choices in the transition function.
For every minimal satisfying assignment $X'$ of $\delta(q)$, i.e., the sets of successor states, at least one state from $X'$ must be included in the chosen set $X$.
Moreover, each literal of each satisfying assignment must be part of the vertices.
This ensures that all nondeterministic branches are unfolded and, whenever more than one state appears in the vertices for the same timepoint, these states are reached via disjunctions.
The key difference from existential unfoldings is that universal unfoldings track states that can be reached through \emph{different nondeterministic choices}, capturing scenarios where distinct runs may allow different valuations at the same timestep -- this is the situation where state-wise universal quantification would yield $\false$ spuriously.
\Cref{fig:universal:unfolding} sketches the universal unfoldings of the automaton in \Cref{fig:running:example:universal}.
Using universal unfoldings, we can define the syntactic condition for universal quantification:

\begin{definition}[Universal Syntactic Condition]\label{def:syntactic:univ:condition}
An alternating automaton $\awa$ is in universal normal form for $p \in \Sigma$ if for all its universal unfoldings $(V, E)$ and for all $i \in \mathbb{N}$ it holds that either $(p, i) \notin V$ or $(\neg p, i) \notin V$.
\end{definition}

\Cref{fig:universal:unfolding} sketches the universal unfolding of the automaton from \Cref{fig:running:example:universal}, which does not satisfy the syntactic condition.
The universal unfolding resolves \emph{universal} choices rather than nondeterministic ones.
At each timestep, for each state $q$, we select a set $X$ of successor states that covers all nondeterministic branches: For every minimal satisfying assignment $X'$ of $\delta(q)$, at least one state from $X'$ must be included in $X$.
In the example, we reach $q_1$ and $q_2$ at timepoint $1$ and find the conflicting literals $a$ and $\neg a$ at timepoint $2$.
Note that we do not have $q_0, 1$ in the unfolding since we follow one state of each conjunction and the transition function of $q_0$ is $(b \wedge q_0 \wedge q_1) \vee (b \wedge q_0 \wedge q_2)$.

We close this section by showing that the syntactic conditions defined in this section imply the normal forms defined in the previous section (precisely, \Cref{def:syntactic:exists:condition} implies \Cref{def:existential:knowledge:compilation}, and~\Cref{def:syntactic:univ:condition} implies~\Cref{def:universal:knowledge:compilation}, respectively). 

\begin{theorem}
Let $\awa$ be an alternating automaton. 
If $\awa$ satisfies the existential syntactic condition, then $\awa$ is in existential normal form.
If $\awa$ satisfies the universal syntactic condition, then $\awa$ is in universal normal form.
\end{theorem}
\begin{proof}
We prove the existential case; the universal case is analogous.
Suppose $\awa$ is an alternating automaton on input alphabet $\Sigma$ that satisfies the condition of \Cref{def:syntactic:exists:condition}.  To establish \Cref{def:existential:knowledge:compilation}, we must show that $\mathcal{L}(\exists^\circ p. \awa)$ and $\mathcal{L}(\exists p. \awa)$ coincide.  
Recall that $\exists^\circ p. \awa$ is simply the automaton $\awa$ in which the transition function of every state $q$ has been modified from $\delta(q)$ to $\exists p. \delta(q)$.  
It follows that for any alternating automaton $\awa$, we always have $\mathcal{L}(\exists p. \awa) \subseteq \mathcal{L}(\exists^\circ p.\awa)$: For every word $w \in (\Sigma \setminus p)^\omega$ that belongs to $\mathcal{L}(\exists p.~\awa)$, there exists a sequence of assignments to $p$, say $a_p$, corresponding to an accepting run of $\awa$ on $w \sqcup a_p$, i.e., $w$ augmented with $a_p$.  
The $j$\textsuperscript{th} element of this sequence, namely $a_p[j]$, serves as a witness assignment to $p$ in $\exists p. \delta(q)$ for every state $q$ of $\awa$ visited in the $j$\textsuperscript{th} step in the same accepting run of $w \sqcup a_p$ on $\awa$.  
Hence, the word $w$ is in $\mathcal{L}
(\exists^\circ p. \awa)$.
To complete the proof, it remains to show that $\mathcal{L}(\exists^\circ p. \awa) \subseteq \mathcal{L}(\exists p. \awa)$.
This is where \Cref{def:syntactic:exists:condition} is used.

Let $w \in (\Sigma \setminus p)^\omega$ be a word in $\mathcal{L}(\exists^\circ p. \awa)$.
Consider an accepting run of the automaton $\exists^\circ p. \awa$ on $w$.
At every step $j$ of this run, the transition formula $\exists p. \delta(q)$ is satisfied for every state $q$ that appears at the $j$\textsuperscript{th} step of the run.  In general, the witness for $p$ used to satisfy $\exists p. \delta(q)$ at step $j$ may vary with the state $q$. 
However, since $\awa$ satisfies the condition of \Cref{def:syntactic:exists:condition}, we know that in every existential unfolding of $\awa$, there is no level $j$ at which both $p$ and $\neg p$ appear.
This allows us to construct a sequence $a_p$ of Boolean assignments to $p$, such that $a_p[j]$ serves as the common witness for $p$ in $\exists p. \delta(q)$ for every state $q$ appearing at the $j$\textsuperscript{th} step of the above accepting run of $\exists^\circ p. \awa$.
Specifically, if only $p$ appears in the $j$\textsuperscript{th} step of the existential unfolding corresponding to the accepting run, we set $a_p[j] = 1$; otherwise, we set $a_p[j] = 0$.
Since the same value $a_p[j]$ serves as the witness for $p \in \exists p. \delta(q)$ for all states $q$ appearing at the $j$\textsuperscript{th} step of the run, it follows that the sequence $a_p$ is a witness sequence of Boolean assignments to $p$ in an accepting run of $\awa$ on $w \sqcup a_p$. 
This implies $w$ is in $\mathcal{L}(\exists p. \awa)$.
\end{proof}

The syntactic conditions are sufficient for the automaton to be in normal form.
In the following section, we develop algorithms for identifying violations of the syntactic conditions and how to fix them, allowing us to transform automata into their normal forms. 

\section{Knowledge Compilation Algorithms for Safety Automata}\label{sec:computing:knowledge:compilations:for:safety:automata}
In the following, we focus on alternating safety automata.
While the definitions in \Cref{sec:knowledge:compilations:for:existential:and:universal:quantification} apply to all acceptance conditions, the following algorithms operate on the transition structure of the automata, not the acceptance condition, making safety the most suitable choice.
The main idea for computing existential and universal knowledge compilations is twofold: We first identify a set of states that need to be conjunctively or disjunctively merged in the automaton, and then perform a subset construction of the automaton for the identified states.
We introduce \emph{conflict states} for the first part and present fixpoint algorithms over the transition function of an automaton to iteratively collect conflicting states in the automaton.
The second part is two modified versions of the Miyano-Hayashi construction~\cite{DBLP:journals/tcs/MiyanoH84}, a construction that translates alternating to nondeterministic automata, for only a subset of states and the safety acceptance condition. 

\subsection{Existential Knowledge Compilation}
We begin with the definition of existentially conflicting states in an automaton.

\begin{definition}[Existential Conflict]\label{def:existential:conflicts}
Two states $q,q' \in \states$ are existentially \emph{in conflict for $p$} if there exists an existential unfolding $(V, E)$ of $\mathcal{A}$ and $i\in \mathbb{N}$ s.t. $(q, i)$ and $(q', i) \in V$ and there exists an $i' \in \mathbb{N}$ s.t.  $(p, i') \in \mathit{reach}(q,i)$ and $(\neg p, i') \in \mathit{reach}(q',i)$.
\end{definition}

The states are in existential conflict if they can be reached at the same timepoint in an existential unfolding, i.e., they are below a conjunction, and there is a future timepoint on which $p$ holds on the path from $q$ and $\neg p$ holds on a path from $q'$.
In the example in \Cref{fig:running:example:exists}, $q_0$ and $q_1$ are in existential conflict: There exists an existential unfolding where $(q_0, 1)$ and $(q_1,1)$ are reached, and $(a, 1)$ is reachable from $(q_0,1)$ and $(\neg a, 1)$ is reachable from $(q_1, 1)$.
We show that we can compute whether two states are in existential conflict.

\begin{restatable}{theorem}{existentialconflict}\label{th:existential:conflict}
Deciding whether two states of an alternating automaton $\awa$ are existentially in conflict is quadratic in the number of states of $\awa$.
\end{restatable}

\begin{proof}
We construct an automaton that decides whether two states are in existential conflict.
Let $\aut = (\Sigma,\states,q_0,\delta,\final)$ be an alternating \Buchi automaton, $q, q' \in \states$ be two states, and $p \in \Sigma$ be a letter of the alphabet.
The nondeterministic safety automaton $\aut'$ is non-empty if $q, q'$ are in conflict, where $\aut' = (\Sigma', \states',q_0',\delta',F')$ with $\Sigma' = \emptyset$ and 
\begin{align*}
    \states' &= Q \times Q \times Bool\\
    q_0' &= (q_0, q_0, \false)\\
    \delta'(q_1, q_2, b) &= \{(q_1', q_2', b') \mid b' := b \vee (q_1 = q \wedge q_2 = q')\\
&~~~~~~~\mathit{ and}~\exists X \subseteq Q \cup L_{AP}~\mathit{s.t.} \\
&~~~~~~~~~~q_1', q_2' \in X, ~ X \vDash \delta(q_1) \mathit{, and}~X' \vDash\delta(q_2)\} \\
    F' &= \{(p, \neg p, \true), (\neg p, p, \true)\} 
\end{align*}
The automaton is quadratic in the number of states in $\mathcal{A}$.
\end{proof}

Note that \Cref{def:existential:conflicts} is only a sufficient condition based on the syntactic structure of the automaton.
We also consider states to be in conflict that are semantically unreachable, i.e., the existential unfolding does not represent a run of the automaton.
However, the condition enables us to effectively compute an over-approximation of the states that are in conflict.

\begin{theorem}\label{th:existential:normal:form}
    Let $\awa = (\Sigma, \states, q_0, \delta, F)$ be an alternating safety automaton and $p \in \Sigma$ a Boolean variable. If $\awa$ has no states that are existentially in conflict for $p$, then $\awa$ is in existential normal form for $p$.
\end{theorem}
\begin{proof}
   We prove that $\lang (\exists p. \awa) = \lang(\sexists p. \awa)$.
   $\lang(\exists p. \awa) \subseteq \lang(\sexists p. \awa)$:
   Let $w \in \lang(\exists p. \awa)$. Then there exists a $w' \in (2^\Sigma)^\omega$ with $w'|_{\Sigma \setminus \{p\}} = w$ and $w' \in \mathcal{L}(\awa)$. Let $(\mathcal{T}, r)$ be the accepting run tree of $\awa$ on $w'$. For every node $n \in \mathcal{T}$, the set $S$ of children of $n$ satisfies $\delta(r(n))$. Recall that the state-wise existential quantification $\sexists p. \awa$ shares the same set of states as $\awa$, but replaces the transition function with $\delta'(r(n)) = \exists p. \delta(r(n))$. Since $\delta(r(n)) \rightarrow \exists p. \delta(r(n))$ is a tautology, and $S$ satisfies $\delta(r(n))$, $S$ also satisfies $\delta'(r(n))$. Thus, $(\mathcal{T}, r)$ is an accepting run tree of $\sexists p. \awa$ on $w$.
   $\lang(\sexists p. \awa) \subseteq \lang(\exists p. \awa)$: We prove this by contradiction. 
   Assume $\lang(\sexists p. \awa) \not\subseteq \lang(\exists p. \awa)$. 
   Thus, there exists a word $w \in \lang(\sexists p. \awa) \setminus \lang(\exists p. \awa)$. 
   Since $w \in \lang(\sexists p. \awa)$, there exists an accepting run tree $(\mathcal{T}, r)$ of $\sexists p. \awa$ on $w$. 
   Also, since $w \not\in \lang(\exists p. \awa)$, and $\sexists p.$ only changes the transition function regarding to $p$, the transition function of one state of one of the branches of $(\mathcal{T}, r)$ evaluated to $\false$, hence rejecting the word $w$. 
   Specifically, there exists a timestep $i \in \mathbb{N}$ where $(\mathcal{T}, r)$ branches conjunctively to states $q_1, q_2$, s.t. for some timestep $i' \geq i$, the branch from $q_1$ requires $p$ at $i'$ and the branch from $q_2$ requires $\neg p$ at $i'$. 
   We then can construct the existential unfolding $(V,E)$ of $\awa$ that follows the run tree $(\mathcal{T}, r)$ by guessing the correct nondeterministic choices in the transition function and following the universal branches.
   Consequently, this existential unfolding contains both $(q_1, i)$ and $(q_2, i)$ and, therefore, $(p, i')$ and $(\neg p, i')$. This is a contradiction to $\awa$ being in existential normal form.
\end{proof}

\begin{algorithm}[t]
            \caption{Existential Fixpoint Algorithm}
            \label{alg:ex-fixpoint}
        \begin{mycode}
let existentialConfs($\mathcal{A} = (\Sigma, Q, q_0, \delta, F), p \in \Sigma)$ := 
@@ let confs = $\emptyset$
@@ for $q \in Q$ do $l[q, 0]$ = $\{l(q) \cap L_{\{p\}}\}$ 
@@ for $ i = 1; i$++ do 
@@@@ for $q \in Q$ do
@@@@@@ $l[q, i] =  \bigcup_{X \in post(q)} \{ \bigcup_{q' \in X} \bigcup_{u \in l[q', i-1]} u\}$  
@@@@ for $q \in Q$ do
@@@@@@ if $\exists X \in post(q). \exists q', q'' \in X.$ 
@@@@@@@@@@@@@@@@@@@@@@@@@@@@@@@@ $l[q', i-1] \neq l[q'', i-1]$ do
@@@@@@@@ confs = confs $\cup~\mathit{reach}(q') \cup \mathit{reach}(q'') \cup \{q', q''\}$
@@@@ if $\exists i'< i. \forall q' \in Q. l[q', i] = l[q', i']$ return confs 
\end{mycode}%
\end{algorithm}

\paragraph{Computing Existential Conflicts.}
We define a fixpoint algorithm (see Alg.~\ref{alg:ex-fixpoint}) over the alternating automaton that computes a set of states $M$ such that the states in $Q \setminus M$ are not in conflict.
The algorithm maintains a labeling function $l$ that maps each state and timestep to a set of sets of literals.
Initially, for each state $q$, we set $l[q, 0]$ to contain only the literals appearing in $q$'s transition function $\delta(q)$.
In each iteration $i$, we compute $l[q, i]$ by considering all minimal satisfying assignments $X$ of $\delta(q)$ (obtained via $post(q)$), and for each successor state $q' \in X$, we take the union of all label sets $u \in l[q', i-1]$ from the previous iteration.
This propagates the reachable literals backward through the automaton's transition structure.
After updating all labels, we check for conflicts: If there exists a state $q$ with two successors $q', q''$ in some minimal satisfying assignment $X$ such that $l[q', i-1] \neq l[q'', i-1]$, then these successors are reached conjunctively but have different reachable literals at the same timepoint, indicating a conflict.
We add both $q'$ and $q''$, and all their successors, to the conflict set.
The algorithm terminates when a fixpoint is reached, i.e., when there exists a previous iteration $i'$ such that $l[q', i] = l[q', i']$ for all states $q'$.
This occurs after all loop combinations in the automaton have been processed.

We demonstrate the core ideas of the algorithm on our running example in \Cref{fig:running:example:exists} with $p = a$. We initialize the labeling function $l$ with the literals of $L_{\{a\}}$ appearing in each state's transition function: $l[q_0, 0] = \{\{a\}\}$ and $l[q_1, 0] = \{\{\neg a\}\}$. The algorithm then enters the main loop (line 5) for $i = 1$. Since $post(q_0) = \{\{q_0, q_1\}\}$, it sets $l[q_0, 1] = \{l[q_0, 0] \cup l[q_1, 0]\} = \{\{\neg a, a\}\}$. The state $q_1$ has no successor states, so $l[q_1, 1] = \emptyset$. Proceeding to the conflict detection (line 8), the algorithm checks if any two states in the conjunction $X = \{q_0, q_1\}$ have differing labels in the previous step $(i = 0)$. Since $l[q_0, 0] = \{\{a\}\} \neq \{\{\neg a\}\} = l[q_1, 0]$, the states $q_0$ and $q_1$ are added to the conflict set. In the next iteration of the loop $(i = 2)$, the labeling remains the same, i.e., $l[q_0, 2] = \{\{\neg a, a\}\}$ and $l[q_1, 2] = \emptyset$. Both states would once again be recognized as conflicts for the previous iteration $(i = 1)$. The termination condition (line 11) evaluates to $\true$ since $l[q_0, 1] = l[q_0, 2]$ and $l[q_1, 1] = l[q_1, 2]$, returning the conflict set $\{q_0, q_1\}$.

\paragraph{Existential Compilation of Alternating Safety Automata.}

We adapt the Miyano-Hayashi construction~\cite{DBLP:journals/tcs/MiyanoH84} for knowledge compilation for a set of conflicting states.
We use a subset construction to define the set of conflicting states $M$.
The set $M$ is such that for no $q, q' \in \states \setminus M$ it holds that $q$ is in existential conflict with $q'$.
We convert the transition function of each state to a conjunction of conflict states combined with non-conflict states as defined in the original transition function.
For every conjunction of states, we build the conjunction of their transition functions and proceed to the resulting successor states.

\begin{theorem}\label{th:existential:knowledge:compilation}
Let $\mathcal{A} = (\Sigma, \states, q_0, \delta, F)$ be an alternating safety automaton, and let $M \subseteq Q$ be an over-approximation of the existential conflicts in $Q$ for $p \in \Sigma$.
There exists an alternating automaton of exponential size in $M$ that is in existential normal form for $p$.
\end{theorem}

\begin{proof}
We construct $\mathcal{A}' = (\Sigma', \states', q_0', \delta', F')$ as follows, where $\Sigma' = \Sigma$ and $F = Q'$.
\begin{align*}
\states' &= (2^M) \cup (\states \setminus M),
~~ q_0' = \twopartelse{\{q_0\}}{q_0 \in M}{q_0}\\
  \delta'(q) &= \bigvee_{X' \subseteq M} X' \wedge  \delta(q)[X' \rightarrow 1, M \setminus X' \mapsto 0]  \\
  \delta'(X) &= \bigvee_{X' \subseteq M} X'
  \wedge\, \bigwedge_{q\in X} \delta(q)[X' \mapsto 1, M \setminus X' \mapsto 0]
\end{align*}
In this alternating automaton, there exists no accepting run on which two states $q, q' \in M$ appear on different paths at the same timepoint.
It is, therefore, in existential normal form.
The automaton is exponential in the size of $M$.
\end{proof}

\subsection{Universal Knowledge Compilation}
Universal knowledge compilations change the perspective on conflicting states: We must identify pairs of states on which the universal quantification of a variable in a single state of the automaton would lead to a spurious result.
Assume there exist two different runs in an alternating automaton and state $q$ is reached in the first and $q'$ in the second.
Furthermore, $\delta(q)$ enforces that $p$ holds, but $\delta(q')$ enforces that $\neg p$ holds on the input word.
Universal quantification should evaluate to $\true$ in this case, since for any evaluation of $p$, there exists a satisfying assignment, one in $\delta(q)$ and one in $\delta(q)$. 
However, if \Cref{def:state:wise:universal:projection} is applied to this automaton, $\delta(q)$ and $\delta(q')$ will both evaluate to $\false$, spuriously rejecting all words that traverse through $q$ and $q'$.
We solve this by, again, combining both states in the knowledge compilation.
In contrast to the existential compilation, we consider conflicting states to be in \emph{disjunction}.
We begin with the definition of a universal conflict.

\begin{definition}[Universal Conflict]\label{def:universal:conflicts}
Two states $q,q' \in \states$ are universally \emph{in conflict for $p$} if there exists a universal unfolding $(V, E)$ of $\mathcal{A}$ and $i\in \mathbb{N}$ s.t. $(q, i)$ and $(q', i) \in V$ and there exists an $i' \in \mathbb{N}$ s.t.  $(p, i') \in \mathit{reach}(q,i)$ and $(\neg p, i') \in \mathit{reach}(q',i)$.
\end{definition}

We ensure that states in universal conflict are preceded by a nondeterministic choice in a state of the automaton, which guarantees that there exist two different accepting runs on which the states are reached at the same timestep.
Intuitively, this is the case if the automaton contains a nondeterministic choice that leads to different parts of the automaton.
Checking whether two states are in universal conflict is quadratic in the number of states.

\begin{restatable}{theorem}{universalconflict}\label{th:universal:conflict}
Deciding whether two states of an alternating automaton $\awa$ are universally in conflict is quadratic in the number of states of $\awa$.
\end{restatable}
\begin{proof}
We construct an automaton that decides whether two states are in universal conflict.
Let $\aut = (\Sigma,\states,q_0,\delta,\final)$ be an alternating \Buchi automaton, $q, q' \in \states$ be two states, and $p \in \Sigma$ be a letter of the alphabet.
The nondeterministic safety automaton $\aut'$ is non-empty if $q, q'$ are in conflict, where $\aut' = (\Sigma', \states',q_0',\delta',F')$ with $\Sigma' = \emptyset$ and 
\begin{align*}
    \states' &= Q \times Q \times Bool\\
    q_0' &= (q_0, q_0, \false)\\
    \delta'(q_1, q_2, b) &= \{(q_1', q_2', b') \mid b' := b \vee (q_1 = q \wedge q_2 = q')\\
&~~~~~~\mathit{ and}~\exists X, X' \subseteq Q \cup L_{AP}~\mathit{s.t.}~q_1' \in X,\\
&~~~~~~~~~~~~ q_2' \in X', X \vDash \delta(q_1) \mathit{, and}~ X' \vDash\delta(q_2)\} \\
    F' &= \{(p, \neg p, \true), (\neg p, p, \true)\} 
\end{align*}
The automaton is quadratic in the number of states in $\mathcal{A}$.
\end{proof}

\begin{restatable}{theorem}{universalconflictfree}\label{thm:univ:conflictfree}
    Let $\awa = (\Sigma, \states, q_0, \delta, F)$ be an alternating safety automaton and $p \in \Sigma$ a Boolean variable. If $\awa$ has no states that are universally in conflict for $p$, then $\awa$ is in universal normal form for $p$.
\end{restatable}

\begin{proof}
    We prove that $\lang(\forall p. \awa) = \lang(\sforall p. \awa)$.
    $\lang(\forall p. \awa) \subseteq \lang(\sforall p. \awa)$: We prove this by contradiction. Assume $\lang(\forall p. \awa) \not\subseteq \lang(\sforall p. \awa)$. Thus, there exists a word $w \in \lang(\forall p. \awa) \setminus \lang(\sforall p. \awa)$. Since $w \in \lang(\forall p. \awa)$, for every $w' \in (2^\Sigma)^\omega$ with $w'|_{\Sigma \setminus \{p\}} = w$, there exists an accepting run tree of $\awa$ on $w'$. Also, since $w \not\in \lang(\sforall p. \awa)$, and $\sforall p.$ only changes the transition function regarding to $p$, the transition function $\delta'$ of one state of one of the branches of any single potential run tree evaluated to $\false$, hence rejecting the word $w$.
    Specifically, there exists a timestep $i$ where the run trees branch disjunctively to states $q_1, q_2$, s.t. for some timestep $i' \geq i$, the branch from $q_1$ requires $p$ at $i'$ and the branch from $q_2$ requires $\neg p$ at $i'$.
    If this is the case, then we can construct the universal unfolding $(V, E)$ of $\awa$ that captures the disjunctive branching of $\awa$ utilized by the run trees. This universal symbolic unfolding contains both $(q_1, i)$ and $(q_2, i)$, and, therefore, $(p, i)$ and $(\neg p, i)$. This is a contradiction to $\awa$ being in universal normal form.

    $\lang(\sforall p. \awa) \subseteq \lang(\forall p. \awa)$: Let $w \in \lang(\sforall p. \awa)$. We show that for all $w' \in (2^\Sigma)^\omega$ with $w'|_{\Sigma \setminus \{p\}} = w$, it holds that $w' \in \lang(\awa)$. Let $(\mathcal{T}, r)$ be the accepting run tree of $\sforall p. \awa$ on $w$. For every node $n \in \mathcal{T}$, the set $S$ of children of $n$ satisfies $\delta'(r(n)) = \forall p$. Recall that the state-wise universal quantification $\sforall p. \awa$ shares the same set of states as $\awa$, but replaces the transition function with $\delta'r(n)) = \forall p. \delta(r(n))$. At any timestep, the word $w'$ assigns a value $v \in \{0, 1\}$ to $p$. Since $\forall p. \delta(r(n)) \rightarrow \delta(r(n))$ is a tautology, and $S$ satisfies $\delta'(r(n))$, $S$ also satisfies $\delta(r(n))$ when $p$ is assigned the value $v$. Thus, $(\mathcal{T}, r)$ is an accepting run tree of $\awa$ on $w'$. Because this holds for all such $w'$, $w \in \lang(\forall p. \awa)$.
\end{proof}

\paragraph{Computing Universal Conflicts.}
The universal conflict fixpoint algorithm (see Alg.~\ref{alg:univ-fixpoint}) follows a similar structure to the existential case, but with a crucial difference in how conflicts are detected.
As before, we maintain a labeling function $l$ that maps each state and timestep to sets of literals, initialized with $l[q, 0] = \{l(q) \cap L_{\{p\}}\}$ for each state $q$.
The key difference appears in line 7: Instead of checking for conflicts within a single minimal satisfying assignment $X$ (which corresponds to conjunctive branching), we check for conflicts \emph{across different} minimal satisfying assignments $X$ and $X'$ of $\delta(q)$.
This reflects the nondeterministic nature of universal conflicts: The automaton can choose different paths (different satisfying assignments) that impose contradictory requirements on the valuation of $p$.
Specifically, if there exist distinct minimal satisfying assignments $X, X' \in post(q)$ with successors $q' \in X$ and $q'' \in X'$ such that $l[q', i-1] \neq l[q'', i-1]$, we identify a universal conflict.
This captures the scenario where different runs of the automaton, arising from nondeterministic choices, reach states with incompatible restrictions on the value of $p$ at the same timestep.
We add the conflicting states and their successors to the conflict set and terminate when a fixpoint is reached, which occurs exactly after traversing all loop combinations in the automaton.

\begin{algorithm}[t]
            \caption{Universal Fixpoint Algorithm}
            \label{alg:univ-fixpoint}
        \begin{mycode}
let universalConfs($\mathcal{A} = (\Sigma, Q, q_0, \delta, F), p \in \Sigma)$ := 
@@ let confs = $\emptyset$
@@ for $q \in Q$ do $l[q, 0]$ = $\{l(q) \cap L_{\{p\}}\}$ 
@@ for $ i = 1; i$++ do 
@@@@ for $q \in Q$ do
@@@@@@ $l[q, i] =  \{\bigcup_{X \in post(q)} \bigcup_{q' \in X} \bigcup_{u \in l[q', i-1]} u\}$  
@@@@ for $q \in Q$ do
@@@@@@ if $\exists X, X' \in post(q). \exists q' \in X. \exists q'' \in X'.$ 
@@@@@@@@@@@@@@@@@@@@@@@@@@@@@@@@@@ $l[q', i-1] \neq l[q'', i-1]$ do
@@@@@@@@ confs = confs $\cup~\mathit{reach}(q') \cup \mathit{reach}(q'') \cup \{q', q''\}$
@@@@ if $\exists i'< i. \forall q' \in Q. l[q', i] = l[q', i']$ return confs
\end{mycode}%
\end{algorithm}

We demonstrate Alg.~\ref{alg:univ-fixpoint} on the example in \Cref{fig:running:example:universal} where $p = a$. 
We initialize the labeling function $l$ with the literals of $L_{\{a\}}$ appearing in each state's transition function: $l[q_0, 0] = \{\{\neg a, a\}\}$, $l[q_1, 0] = \{\{a\}\}$ and $l[q_2, 0] = \{\{\neg a\}\}$. 
The algorithm then enters the main loop (line 4) for $i = 1$. Since $post(q_0) = \{\{q_0, q_1\}, \{q_0, q_2\}\}$, it sets $l[q_0, 1] = \{l[q_0, 0] \cup l[q_1, 0] \cup l[q_2, 0]\} = \{\{\neg a, a\}\}$. The states $q_1$ and $q_2$ have no successor states, so $l[q_1, 1] = l[q_2, 1] = \emptyset$. 
Proceeding to the conflict detection (line 7), the algorithm checks if any two states across two disjunctive branches had differing labels in the previous step $(i = 0)$. 
Note that the structure of $post(q_0)$ tells us precisely which states can be reached conjunctively or disjunctively: Because $q_1$ and $q_2$ are in different \emph{inner} sets of $post(q_0)$, we know they are reached disjunctively. 
Since $l[q_1, 0] = \{\{a\}\} \neq \{\{\neg a\}\} = l[q_2, 0]$, the states $q_1$ and $q_2$ are added to the conflict set. Analogously, $q_0$ is added to the conflict set. 
In the next iteration of the loop $(i = 2)$, the labeling remains the same, i.e., $l[q_0, 2] = \{\{\neg a, a\}\}$ and $l[q_1, 2] = l[q_2, 2] = \emptyset$.
No new conflicts are detected for the previous iteration $(i = 1)$. The termination condition (line 11) evaluates to $\true$ since $l[q_0, 1] = l[q_0, 2]$, $l[q_1, 1] = l[q_1, 2]$ and $l[q_2, 1] = l[q_2, 2]$, returning the conflict set $\{q_0, q_1, q_2\}$.
Note that this is a superset of the conflicts shown in \Cref{fig:running:example:universal}.

\paragraph{Universal Compilation of Alternating Safety Automata.}
We now present an automaton construction that, given a set of conflicting states, produces an alternating automaton where these states are disjunctively combined whenever they are reached simultaneously.

\begin{theorem}\label{th:universal:knowledge:compilation}
Let $\mathcal{A} = (\Sigma, \states, q_0, \delta, F)$ be an alternating safety automaton, and let $M \subseteq Q$ be an over-approximation of the universal conflicts in $Q$ for $p \in \Sigma$.
There exists an alternating automaton of exponential size in $M$ that is in universal normal form for $p$.
\end{theorem}

\begin{proof}
We construct $\mathcal{A}' = (\Sigma', \states', q_0', \delta', F')$ as follows, where $\Sigma' = \Sigma$ and $F = Q'$.
\begin{align*}
\states' &= (2^M) \cup (\states \setminus M), ~~
q_0' = \twopartelse{\{q_0\}}{q_0 \in M}{q_0}\\
  \delta'(q) &= \bigwedge_{X' \subseteq M} X' \wedge \delta(q)[X' \mapsto 1, M \setminus X' \mapsto 0]  \\
  \delta'(X) &= \bigwedge_{X'\subseteq M} X' \wedge\, \bigvee_{q\in X} \delta(q)[X' \mapsto 1, M \setminus X' \mapsto 0]
\end{align*}
In this alternating automaton, there are no two accepting runs for the same word on which $q, q' \in M$ appear at the same timepoint.
It is, therefore, in universal normal form.
The automaton is exponential in the size of $M$.
\end{proof}

The construction operates as follows:
For each non-conflicting state $q \in \states \setminus M$, we replace transitions to conflicting states with transitions to their corresponding subset representation. 
For each subset $X \subseteq M$ of conflicting states, we construct a new macro-state in $\states'$ whose transition function combines the original transitions disjunctively.
The key observation distinguishing this construction from the existential knowledge compilation is the use of \emph{universal} edges (conjunctions) at the top level of $\delta'(q)$ and $\delta'(X)$, in contrast to the nondeterministic edges (disjunctions) used in both the existential compilation and the original Miyano-Hayashi construction~\cite{DBLP:journals/tcs/MiyanoH84}. 
This ensures that when multiple conflicting states would be reached simultaneously, they are instead grouped into a single macro-state that disjunctively combines their transition function and the future reachable states, thereby eliminating the universal conflict while preserving the language of the automaton.

\begin{corollary}[Knowledge Compilation Framework]
\label{cor:unified-framework}
Let $\mathcal{A} = (\Sigma, \states, q_0, \delta, F)$ be an alternating safety automaton and $p \in \Sigma$ be a Boolean variable.
\begin{enumerate}
    \item For existential quantification, we first compute the existential conflicts $M_\exists$ with Alg.~\ref{alg:ex-fixpoint}, then apply the existential compilation (cf.~\Cref{th:existential:knowledge:compilation}) to compute $\mathcal{A}_\exists$, and finally perform state-wise quantification $\mathcal{A}'_\exists = \sexists p. \mathcal{A}_\exists$. It holds that $\mathcal{L}(\mathcal{A}'_\exists) = \mathcal{L}(\exists p. \mathcal{A})$.
    \item For universal quantification, we first compute the universal conflicts $M_\forall$ with Alg.~\ref{alg:univ-fixpoint}, then apply the universal compilation (cf.~\Cref{th:universal:knowledge:compilation}) to compute $\mathcal{A}_\forall$, and finally perform state-wise quantification $\mathcal{A}'_\forall = \sforall p. \mathcal{A}_\forall$. It holds that $\mathcal{L}(\mathcal{A}'_\forall) = \mathcal{L}(\forall p. \mathcal{A})$.
\end{enumerate}
\end{corollary}

\section{Experiments}
We implemented our knowledge compilation approach in a tool called \tool{}~\cite{metzger_2026_20009614}, written in F\# and built on top of several existing libraries and tools.
The implementation uses the C\# BDD package\footnote{\href{https://github.com/microsoft/DecisionDiagrams}{https://github.com/microsoft/DecisionDiagrams}} for efficient manipulation of Boolean formulas and the FsOmegaLib library\footnote{\href{https://github.com/ravenbeutner/FsOmegaLib}{https://github.com/ravenbeutner/FsOmegaLib}} for automaton operations.
For LTL-to-automaton translation and emptiness checking, we leverage \textsc{SPOT}~\cite{DBLP:conf/cav/Duret-LutzRCRAS22} as a backend.
All experiments are performed on an Apple M3 Max with 48 GB RAM. We set the per-run timeout to 30 seconds to make the experimental evaluation tractable.
.

\subsection{Implementation}
A key design decision in our implementation is a specialized automaton representation that directly reflects the formal definition of alternating automata.
In particular, we represent the transition function $\delta(q)$ of each state $q$ as a single BDD that encodes both the constraints on the input alphabet and the successor state assignments.
The BDD for $\delta(q)$ is constructed over a variable set that includes both the alphabet variables $\Sigma$ and Boolean variables representing states in $Q$.
A satisfying assignment of this BDD simultaneously determines which input letter is read and which successor states are activated.
Concretely, given an alphabet $\Sigma = \{a, b\}$ and states $Q = \{q_0, q_1, q_2\}$, the BDD for a transition function is built over the variables $\{a, b, q_0, q_1, q_2\}$.
From this BDD representation, we can efficiently extract the minimal satisfying assignments $\minsat(\delta(q))$ used in our fixpoint algorithms.
Each minimal satisfying assignment corresponds to a minimal conjunction of literals that satisfies the transition function.
This representation is particularly well-suited for state-wise quantification: Given a variable $p \in \Sigma$ to project, we can directly apply the BDD operations $\exists p$ or $\forall p$ to $\delta(q)$ for each state $q$ independently.
Our implementation provides functionality for both universal and existential knowledge compilation, allowing us to handle arbitrary QPTL formulas with quantifier alternations.

\subsection{Experimental Evaluation}
We use QPTL satisfiability as a case study to obtain a preliminary evaluation of the effectiveness of our knowledge compilation approach.
Let 
$ \varphi = \mathbb{Q}_0 p_0. \ldots \mathbb{Q}_n p_n.\varphi'$
be a QPTL formula where $\mathbb{Q} \in \{\exists, \forall\}$ and $\varphi'$ is the quantifier-free body of the formula.
We implemented QPTL satisfiability by applying the results of \Cref{cor:unified-framework} from $\mathbb{Q}_n p_n$ to $\mathbb{Q}_0 p_0$.
The starting point is the automaton for $\varphi'$, constructed by \textsc{SPOT}.
For each $\mathbb{Q} p$, we apply \Cref{cor:unified-framework}(1) if $\mathbb{Q} = \exists$ or \Cref{cor:unified-framework}(2) if $\mathbb{Q} = \forall$.
After all $\mathbb{Q}_i p_i$ are projected, we check the resulting automaton with \textsc{SPOT} for emptiness.
The formula $\varphi$ is satisfiable if \textsc{SPOT} returns that the automaton is nonempty.
\begin{figure}[t]
\includegraphics[width =.78\linewidth]{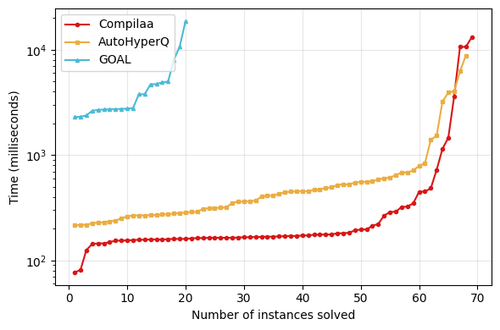}
\centering
\caption{A cactus plot comparing \tool{} with \autohyperq{} and \goal{} on SYNTCOMP benchmarks.}
\label{fig:cactus}
\end{figure}

\paragraph{Evaluation.}
We evaluated our approach on benchmarks derived from the latest edition of the synthesis competition (SYNTCOMP)~\cite{DBLP:journals/corr/abs-2206-00251}.\footnote{\hyperlink{https://github.com/SYNTCOMP/benchmarks}{https://github.com/SYNTCOMP/benchmarks}}
These benchmarks originate from synthesis problems, which we transform to \emph{universal synthesis} problems expressed as QPTL satisfiability problems.
A universal synthesis problem asks whether there are infinite evaluations of output traces for all valuations of environment traces, such that a given temporal specification is satisfied.
Note that this is not the realizability problem, where a finite state strategy must be found that produces correct output sequences on the finite past of input traces.
We produce QPTL formulas of the form $\forall \vec{i}. \exists \vec{o}. \varphi$, where $\vec{i}$ are the input variables, $\vec{i}$ are the output variables, and $\varphi$ is an LTL formula expressing the specification.
We converted all safety SYNTCOMP benchmarks to this QPTL format and compared our tool \tool{} to the state-of-the-art QPTL satisfiability solvers \autohyperq{}~\cite{DBLP:conf/lpar/BeutnerF23} and \goal{}~\cite{DBLP:conf/cav/TsaiTH13}.
In a second set of benchmarks, we evaluated our tool \tool{} on QPTL formulas with quantifier alternations per variable. We use \textsc{SPOT}'s \textsc{randltl} to generate 150 random LTL formulas and then convert them to QPTL formulas by assigning quantifiers for each variable in an alternating way. We check the resulting QPTL formulas for satisfiability and compare the execution times of \tool{} and \autohyperq{}. 
Note that both \textsc{randltl} and the safety specifications of SYNTCOMP limit the scalability of this evaluation.
In future work, we imagine that an evaluation on model-checking benchmarks, where \tool{} is used as a backend, will further highlight the effectiveness of the knowledge compilation algorithms.

\paragraph{Results.}
\begin{figure}[t]
\includegraphics[width =.78\linewidth]{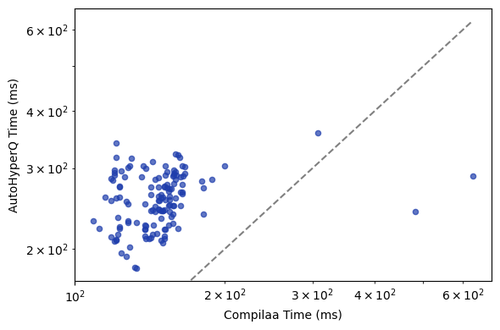}
\centering
\caption{A scatter plot comparing the runtimes of \tool{} and \autohyperq{} on 150 randomly generated QPTL formulas.}
\label{fig:scatter}
\end{figure}
\Cref{fig:cactus} presents a cactus plot comparing the performance of \tool{} against \autohyperq{} and \goal{} on the SYNTCOMP benchmarks.
The x-axis shows the number of solved instances, while the y-axis represents the time taken in seconds.
Each point on a curve represents one benchmark instance, sorted by solving time.
Our tool \tool{} performs competitively with state-of-the-art solvers.
While \goal{} runs into timeouts quickly, \tool{} and \autohyperq{} have comparable runtimes for most benchmarks.
However, as problem difficulty increases (moving right), \tool{} shows improved scalability: It solves almost all instances faster and maintains lower solving times on harder benchmarks.
The performance advantage of \tool{} is clearer with more quantifier alternation depth.
Even though the automata for the randomly generated formulas are small, capped by the execution times of \textsc{randltl}, we can observe a clear performance increase of \tool{} compared to \autohyperq{} in \Cref{fig:scatter}. 
This validates our key hypothesis: By avoiding expensive complementation operations through direct knowledge compilation on alternating automata, we can handle quantifications on variables in alternating automata more efficiently.
The BDD-based representation and our specialized fixpoint algorithms for conflict detection enable \tool{} to scale better than approaches that rely on full nondeterminization and repeated complementation steps.

\section{Conclusion}
We have presented a novel knowledge compilation framework for both existential and universal quantification in alternating safety automata. By compiling automata into existential and universal normal forms, our approach enables efficient state-wise quantification, thereby avoiding the computational bottleneck of repeated complementation used in standard approaches with nondeterministic automata.
We provided fixpoint algorithms for detecting quantification conflicts and presented constructions to resolve them locally. The experimental evaluation demonstrates that our tool \tool{} outperforms existing tools on QPTL satisfiability benchmarks.
Future work includes extending the algorithms to Büchi acceptance, where repair for universal quantifiers is hard, and using knowledge compilation as a backend for scalable HyperLTL model checkers.

\section*{Acknowledgements}
This work originated out of initial discussions in a workshop held as part of the Extended Reunion for Theoretical Foundations of Computer Systems at Simons Institute for the Theory of Computing, Berkeley during July-August 2024. This work was partially supported by the DFG in project 389792660 (TRR 248 – CPEC) and funded by the European Union through ERC Grant HYPER (No. 101055412).
Views and opinions expressed are however those of the authors only and do not necessarily reflect those of the European Union or the European Research Council Executive Agency.
Neither the European Union nor the granting authority can be held responsible for them.

\bibliographystyle{kr}
\bibliography{bibliography}

\clearpage
\appendix

\end{document}